\documentclass[11pt]{amsart}
\usepackage{amsmath,amssymb, mathrsfs, braket, mathtools, color, soul, cite, booktabs}
\usepackage[capitalize]{cleveref}

\newtheorem{theorem}{Theorem}

\newtheorem{lemma}{Lemma}

\DeclareMathOperator*{\argmax}{argmax}

\crefformat{equation}{(#2#1#3)}

\numberwithin{equation}{section}

\begin{document}

\title[Phase matching
]{Phase Matching for a Generalized Grover's Algorithm
}
\author{
Chris Cardullo, Min Kang}

\date{\today}
\subjclass{Primary: 68Q12 81P68; Secondary: 68W01}
\keywords{quantum computing, generalized Grover's algorithm, optimization phase change, phase matching}

\begin{abstract}
We study the fully generalized Grover's algorithm to find the optimal phase changes for each step of the iteration to maximize gain in probability of observation of the target, and when phase matching is required. We find that classical Grover's algorithm and phase matching remain optimal until the target probability gets close to 1. However, as the probability of observation approaches 1, the optimal phase changes differ from $\pi$ and no longer observe phase matching. We provide the optimization statement to find the optimal phase changes given the current amplitude vector and the size of the set. To analyze this formula, we approach it from a numerical and analytical perspective. With the analytical perspective we take two separate approaches. First we identify and classify critical points under reasonable assumptions on the amplitude vector. Second we work with specific critical points to identify when they are local maxima. Finally, we provide an example of a 5 qubit system and show that for the final iteration the optimal phase changes differ from traditional Grover's algorithm and do not observe strict phase matching, but lead to an increase in the probability of the target.
\end{abstract}

\maketitle

\section{{\bf Introduction}}

Grover's algorithm is a foundational algorithm in the field of quantum computing due to a clear advantage over unstructured search in classical computing \cite{Gro1}. The algorithm leverages Householder reflections to boost the amplitude of the target element so that the probability of observing it is increased. The speedup Grover's algorithm provides is quadratic, and it has been shown that the best quantum unstructured search can achieve is a quadratic speedup over classical \cite{BBHT}. Despite this, there has been continued interest in the study of Grover's algorithm. Grover's algorithm has seen use in quantum machine learning \cite{HRN, NSLWLL}, as well as a subroutine in other quantum algorithms \cite{GWG, GGZW}. Outside of just these use cases, Grover's algorithm has continued to see study in order to adapt it for different settings, as well as to study the impact of different phase changes on its search capabilities.

Generalizations of Grover's algorithm were developed that allowed for arbitrary phase rotations, rather than Householder reflections. This technique was utilized by Long et al \cite{LLZN} coupled with the term phase matching, which at the time, referred to a strict equality between the two phase changes. As research has continued, the term has come to instead refer to a general relationship between the phase changes. Phase matching research has produced methods that maintain high search success across varying proportions of marked and unmarked states \cite{LL, ZWBW, CBD}. It has also seen success in deterministic variants of Grover's algorithm, which aim to achieve an exact search of the target \cite{Ho1}, with recent work by Gu{\c{t}}oiu et al \cite{GTP} rigorously proving the optimality of the phase-matching criterion $\phi=-\psi$ with respect to the number of queries in deterministic quantum search.

In our previous work, we analyzed a generalized Grover's algorithm where we enforced a strict phase matching, and used an arbitrary amplitude vector where the non-target states share the same amplitude. For a Hadamard initialization, we found three distinct regions over the amplitude of the target state where the locally optimal phase change can be described by three cases. One where classical Grover's is optimal, one where a non-trivial phase change is optimal, and one where no phase change is optimal. We found that this second region where a non-trivial phase change is optimal is small, and in most cases, a locally optimized Grover's algorithm with strict phase matching would skip this region and be stuck in the region where no change is optimal. In the case where a Hadamard initialization isn't used, and the starting amplitude vector is allowed to be complex, we find a non-trivial relationship between the locally optimal phase change and the phase of the complex amplitude vector \cite{CK}. In this paper, we will extend this analysis and consider a fully generalized Grover's algorithm with non-strict phase matching.

We start by calculating the target state's amplitude after a single iteration of the fully generalized Grover's algorithm, and simplify as much as possible to obtain an arg max statement that we can work with. This statement can be found in \cref{th:simplifiedargmax}. We then analyzed it both numerically and analytically. For the numerical results, we use global optimization to find the local optimal pairs of phase changes for sets of size $N=2^4$, $N=2^5$, and $N=2^{10}$. In each of these cases, we find that classical Grover's algorithm is locally optimal up until a cutoff point at which the optimal phase changes no longer follow strict phase matching and differ from $\pi$. The specific cutoff is lower for smaller set sizes, and these can be seen in \cref{fig:210phasematch}, \cref{fig:25phasematch}, and \cref{fig:24phasematch}. 

For the analytical results, we take two different approaches. First, we work under the reasonable assumptions that $N$ is large, and either that the probability of observing the target scales appropriately, or that the state is initialized with the Hadamard gate and we are optimizing the first step. With this approach, we find that the local optimal is classical Grover's algorithm. In the second approach, we test two pairs of critical points that correspond to notable examples of optimality found in our previous work. The first is when classical Grover's is optimal, so when both phase changes equal $\pi$. The second is when both phase changes equal 0 which corresponds to the problem of overshooting the target found in classical Grover's algorithm. When both phase changes equal $\pi$, we find a condition on the amplitude of the target state and the size of the set for which these phase changes will be optimal, and as a result, when they are not. When both phase changes are zero, we find that it will never be a local optimum, meaning that when both phase changes are allowed to differ the algorithm will not overshoot the target.

Lastly, we would like to provide an outline for the rest of our paper. In section 2, we will discuss preliminaries on classical Grover's algorithm, and its fully generalized form. In section 3, we will analyze the fully generalized Grover's algorithm and simplify the amplitude after one iteration into a statement that we can optimize. Section 4 will feature our work on finding numerical solutions to the optimal phase changes. Sections 5 and 6 will include our analytical results. Then section 7 will showcase an example of using the locally optimized fully generalized Grover's algorithm on a set of size $N=2^5$. Finally we will discuss our final results and takeaways as well as our future research goals in section 8.

\section{{\bf Preliminaries}}

Consider a finite space $\textbf{S}$ with $|\textbf{S}|=N$, which we will call the search space. Each element is assigned a unique number from 0 to $N-1$, and are identified by the binary representation of their assigned number. In other words, $N=2^n$ with $n$ bits. Let $\ket{\tau}\in\textbf{S}$ be the element we wish to search for, which we will call the target element. Finally, we are equipped with a black box function $f$ such that $f(\ket{\tau})=1$ and $f(\ket{a})=0$ for all $\ket{a}\neq \ket{\tau}$. It is important to note that besides the size of \textbf{S} and the existence of a target element $\ket{\tau}$, we do not have any knowledge on the structure of the space, which we refer to as being unstructured.

Grover's algorithm is the standard unstructured search algorithm used in quantum computing. The general structure of the algorithm can be understood as such. We start with a normalized amplitude state vector, meaning the system is initialized with equal amplitudes assigned to each element using the Hadamard gate. The Grover iterate is then applied to the amplitude vector $\lfloor \frac{\pi}{4}\sqrt{N}\rfloor$ times where $N$ is the size of the search space. This iterative process will increase the amplitude of the target while uniformly reducing the amplitude of the non-target states. The type of process is generally referred to as amplitude amplification. The Grover iterate is constructed by a product of operators, that being $-UI_{\ket{0}}^{\pi}UI_{\ket{\tau}}^{\pi}$, where $U$ is a unitary matrix, and $I_{\ket{a}}^\pi$ is a rotation by $\pi$ of the state $\ket{a}$. The unitary we will use here is the same as the original work by Long \cite{Lo1}, that being the Walsh-Hadamard transform, but it is important to note that any unitary matrix can be used and still maintain a functioning search algorithm \cite{Gro2}.

Since the algorithm can be represented as a series of matrix multiplications, we can determine the amplitude for our target state at any point in the iteration. However, if the size of $\textbf{S}$ is too large, then the resulting matrix multiplications will be difficult to calculate, reducing any benefit gained from this analysis. To that end, we will discuss a method of condensing the amplitude vector for the current state and the matrices to 2-dimensional linear maps regardless of the number of qubits in the system. The method we use for this comes from the work of Long \cite{Lo1}. Due to the nature of the Grover's iterate, every non-target element will experience the same change in amplitude, hence we may as well trace one single element by condensing all of them into a single, normalized element which we define as the following:

\[\ket{a} \equiv \frac{1}{\sqrt{N-1}}\sum_{i\neq \tau} \ket{i}.\]

Now with this, let us state what an arbitrary amplitude vector will be. For \textbf{v} to be a valid amplitude vector its entries must be complex, and the sum of the modulus squared of each entry must equal one. We satisfy these conditions with the following definition of \textbf{v}:

\[\textbf{v} = \begin{bmatrix}
    \sin\alpha e^{i\theta_1}\\
    \cos\alpha e^{i\theta_2}
\end{bmatrix}.\]

Note that in this representation, the amplitude of the target is $e^{i\theta_1}\sin\alpha$ and the amplitude for the remaining non-target representative $\ket{a}$ is $e^{i\theta_2}\cos\alpha$. Because the probability is unaffected by the phase of the amplitude, we can simplify it further.

\begin{equation}\label{eq:phasediff}
    \textbf{v} = \begin{bmatrix}
    e^{i\theta_1}\sin\alpha \\
    e^{i\theta_2}\cos\alpha 
\end{bmatrix} = e^{i\theta_2}\begin{bmatrix}
    e^{i(\theta_1-\theta_2)}\sin\alpha\\
    \cos\alpha
\end{bmatrix},
\end{equation}

\noindent which shares the same probability as

\begin{equation}\label{eq:arbitraryamplitude}
    \textbf{v}^*=\begin{bmatrix}
    e^{i\theta}\sin\alpha \\
    \cos\alpha
\end{bmatrix},
\end{equation}

\noindent where $\theta=\theta_1-\theta_2$. A common way to initialize a state is by using the Hadamard gate. This results in a uniformly distributed amplitude and the amplitude is entirely real. In the 2-dimensional vector view we have just discussed, this will result in an initial amplitude of

\begin{equation}\label{eq:hadamardinitvec}
    \textbf{v} = \begin{bmatrix}
        \sin\alpha\\
        \cos\alpha
    \end{bmatrix} = \begin{bmatrix}
        \frac{1}{\sqrt{N}}\\
        \frac{\sqrt{N-1}}{\sqrt{N}}
    \end{bmatrix}.
\end{equation}

Next, let us look at the condensed linear map that represents the traditional Grover's iterate,

\begin{equation}\label{eq:groveriteratepi}
    A^\pi = \begin{bmatrix}
        e^{i\pi}\Big(\frac{1-e^{i\pi}}{N}-1\Big) & \frac{\sqrt{N-1}}{N}(1-e^{i\pi})\\
        \frac{\sqrt{N-1}}{N}e^{i\pi}(1-e^{i\pi}) & -\Big(\frac{1}{N}+\Big(1-\frac{1}{N}\Big)e^{i\pi}\Big)
    \end{bmatrix}.
\end{equation}

With both \cref{eq:arbitraryamplitude} and \cref{eq:groveriteratepi}, we can now calculate the amplitude at any given step of the iteration process by applying the $A^\pi$ operator that same number of times to $\textbf{v}^*$. Allowing the angle $\pi$ above to be arbitrary, we consider the operator $A^\phi$ with the arbitrary phase change $\phi$ coming from $-UI^{\phi}_{\ket{0}}UI^{\phi}_{\ket{\tau}}$. We can then replace \cref{eq:groveriteratepi} by

\begin{equation}
    A^\phi= \begin{bmatrix}
    e^{i\phi}\Big(\frac{1-e^{i\phi}}{N}-1\Big) & \frac{\sqrt{N-1}}{N}(1-e^{i\phi})\\
    \frac{\sqrt{N-1}}{N}e^{i\phi}(1-e^{i\phi}) & -\Big(\frac{1}{N}+\Big(1-\frac{1}{N}\Big)e^{i\phi}\Big)
\end{bmatrix}.\label{eq:groveriteratearbitrary}\end{equation}

When a phase change that differs from $\pi$ is used, we call this a generalized Grover's algorithm. In our previous work \cite{CK} we found that the optimal phase change at any given step had a non-trivial relationship to the phase of the amplitude vector. This optimal phase change can be found by evaluating the following function

\begin{equation}\begin{split}\label{eq:argmax}
    \argmax_{\phi\in[-\pi,\pi]}  \biggl(&\frac{\sqrt{N-1}}{N}\cos(2\alpha)+\frac{1}{N}\sin(2\alpha)\cos(\phi+\theta)\biggr)(1-\\
    &-\cos(\phi))-\frac{1}{2}\sin(2\alpha)\cos(\phi+\theta),
\end{split}
\end{equation}

\noindent with $\theta$ and $\alpha$ coming from \cref{eq:arbitraryamplitude} and $\theta$ representing the complexity of the current amplitude vector. More precisely, $\theta$ is the phase difference between the two coordinates in the current amplitude vector \textbf{v} as seen in \cref{eq:phasediff}. In a traditional implementation of Grover's algorithm, the complexity will remain zero for the majority of steps, since it is initialized and made entirely real by use of the Hadamard gate. However, after a specific amplitude is achieved, we found that the optimal phase change does differ from $\pi$ even in the real amplitude vector case. 

\section{{\bf Analyzing Fully Generalized Grover's Algorithm}}

Let us start by re-examining the Grover's iterate $-UI_{\ket{0}}^\phi UI_{\ket{\tau}}^\phi$ and its matrix representation $A^\phi$ from \cref{eq:groveriteratearbitrary}. The fully generalized Grover iterate would be $-UI_{\ket{0}}^\psi UI_{\ket{\tau}}^\phi$, and its matrix representation is

\begin{equation}\label{eq:grovitmat}
    A^{\psi,\phi} = \begin{bmatrix}
        e^{i\phi}\bigl(\frac{1-e^{i\psi}}{N}-1\bigr) & \frac{\sqrt{N-1}}{N}(1-e^{i\psi})\\
        \frac{\sqrt{N-1}}{N}e^{i\phi}(1-e^{i\psi}) & -\bigl[\frac{1}{N}+(1-\frac{1}{N})e^{i\psi}\bigr]
    \end{bmatrix}.
\end{equation}

We also will consider a general amplitude vector that we would have at any step in the iteration process. Let $\textbf{v}=[e^{i\theta}\sin\alpha\ , \cos\alpha]^\top$ and $\alpha\in[0,\pi/2]$. With this, we can now calculate the amplitude vector after one iteration.

\[
A^{\psi,\phi}\textbf{v} = \begin{bmatrix}
        e^{i(\theta+\phi)}\bigl(\frac{1-e^{i\psi}}{N}-1\bigr)\sin\alpha+\frac{\sqrt{N-1}}{N}(1-e^{i\psi})\cos\alpha\\
    \\
        \frac{\sqrt{N-1}}{N}e^{i(\theta+\phi)}(1-e^{i\psi})\sin\alpha -\bigl[\frac{1}{N}+(1-\frac{1}{N})e^{i\psi}\bigr]\cos\alpha
\end{bmatrix}.
\]

Since we are trying to optimize the probability of observing the target state $\ket{\tau}$, we only need to focus on the first coordinate of the resulting vector. We will call this $(A^{\psi,\phi}\textbf{v})_1$. Thus we have that $P(\ket{\tau})$, the probability of finding the target $\ket{\tau}$ after an iteration is

\begin{equation}
    \label{eq:targetprobinit}
    \begin{split}
            P(\ket{\tau}) &=\ \biggl|e^{i(\theta+\phi)}\bigl(\frac{1-e^{i\psi}}{N}-1\bigr)\sin\alpha +\frac{\sqrt{N-1}}{N}(1-e^{i\psi})\cos\alpha\biggr|^2
    \end{split}.
\end{equation}

Our goal is to better understand the relationship $\psi$ and $\phi$ have when maximizing this probability.

\begin{lemma}
    \label{lem:phaseadj}
    The target probability \cref{eq:targetprobinit} can be rewritten as 
    \begin{equation}\label{eq:argmaxupdate}
    \begin{split}
        P(\ket{\tau})&=\ \biggl|e^{i\phi}\bigl(\frac{1-e^{i\psi}}{N}-1\bigr)\sin\alpha+\frac{\sqrt{N-1}}{N}(1-e^{i\psi})\cos\alpha\biggr|^2.
    \end{split}
    \end{equation}
\end{lemma}

\begin{proof}
    Observing $(A^{\psi,\phi}\textbf{v})_1$, we note that every instance of $\phi$, one of our phase changes, is summed with $\theta$, the phase given by the original amplitude vector. Using this fact, we see that
    \[\biggl(A^{\psi,\phi}\begin{bmatrix}
    e^{i\theta}\sin\alpha\\
    \cos\alpha
    \end{bmatrix}\biggr)_1=\biggl(A^{\psi,\phi+\theta}\begin{bmatrix}
    \sin\alpha\\
    \cos\alpha\end{bmatrix}\biggr)_1.\]
    We now let $\phi^*=\phi+\theta$. Using this, we can restate \cref{eq:targetprobinit} as
    \begin{equation}\label{eq:probphistar}\begin{split}
        P(\ket{\tau})&=\biggl|e^{i\phi^*}\bigl(\frac{1-e^{i\psi}}{N}-1\bigr)\sin\alpha+\frac{\sqrt{N-1}}{N}(1-e^{i\psi})\cos\alpha\biggr|^2.
    \end{split}
    \end{equation}
    \noindent To simplify notation, we will just use $\phi$ instead of $\phi^*$ through the rest of the paper. When we interpret the optimal $\phi^*$, we will have to recall this and adjust it by $\theta$.
\end{proof}

We will now work with this simplified version to find an argmax statement that can be better analyzed. The solutions to this argmax statement will give us the optimal surface of $\phi$ and $\psi$ values based on fixed $N$ and a range of $\alpha$ values.

\begin{theorem}
    \label{th:simplifiedargmax}
    Suppose we carry out the generalized Grover's algorithm with differing phases as described in \cref{eq:grovitmat}. Then the optimal phase changes at any given step can be found by solving the following optimization
    \begin{equation}\label{eq:argmaxf}
    \begin{split}
                \argmax_{\psi,\phi\in[-\pi,\pi]}&\ \biggl[(N-1)^{3/2}\sin(2\alpha)(\cos(\phi-\psi)-\cos\phi)-2(N-1)\cos(2\alpha)\cos\psi-\\
                & -(N-1)^{1/2}\sin(2\alpha)(\cos(\phi+\psi)-\cos\phi)\biggr].
    \end{split}
    \end{equation}
\end{theorem}
\begin{proof}
    To prove this, let us start with further analyzing \cref{eq:probphistar} with $\phi=\phi^*$. We will first focus on the function inside of the modulus in \cref{eq:probphistar}, which we denote as $f = f(\psi,\phi)$ below.

    \[\begin{split}
            f(\psi,\phi) &= e^{i\phi}\bigl(\frac{1-e^{i\psi}}{N}-1\bigr)\sin\alpha +\frac{\sqrt{N-1}}{N}(1-e^{i\psi})\cos\alpha\\
            & = \frac{1}{N}\biggl[(e^{i\phi}(1-N)-e^{i(\phi+\psi)})\sin\alpha+\sqrt{N-1}(1-e^{i\psi})\cos\alpha \biggr].
    \end{split}
    \]

    \noindent This is a complex valued function, so we will rewrite in the form of $f=Re(f)+iIm(f)$ where

    \[\begin{split}
        Re(f)&=\frac{-1}{N}\biggl[((N-1)\cos\phi+\cos(\phi+\psi))\sin\alpha+\sqrt{N-1}(\cos\psi-1)\cos\alpha \biggr]
    \end{split}\]
    \[\begin{split}
        Im(f) &= \frac{-1}{N}\biggl[((N-1)\sin\phi+\sin(\phi+\psi))\sin\alpha+\sqrt{N-1}\sin\psi\cos\alpha \biggr].
    \end{split}\]

    \noindent We can use this to rewrite the probability as

    \[\begin{split}
        P(\ket{\tau}) &= |f(\psi,\phi)|^2\\
        &=\frac{1}{N^2}\biggl\{\sin^2\alpha\biggl[((N-1)\cos\phi+\cos(\phi+\psi))^2+((N-1)\sin\phi+\sin(\phi+\psi))^2\biggr]+\\
        & +(N-1)\cos^2\alpha\biggl[(\cos\psi-1)^2+\sin^2\psi\biggr] +\\
        &+2\sin\alpha\cos\alpha\sqrt{N-1}\biggl[(\cos\psi-1)((N-1)\cos\phi+\\
        &+\cos(\phi+\psi))+\sin\psi((N-1)\sin\phi+\sin(\phi+\psi))\biggr]\biggr\}.
    \end{split}
    \]

    \noindent We wish to further simplify this. To do so, we consider the following terms properly defined to simplify the expression above, and simplify them each.

    \[\begin{split}
        f_1 &=((N-1)\cos\phi+\cos(\phi+\psi))^2+((N-1)\sin\phi+\sin(\phi+\psi))^2\\
        &= (N-1)^2\cos^2\phi+\cos^2(\phi+\psi)+2(N-1)\cos\phi\cos(\phi+\psi)+\\
        &+(N-1)^2\sin^2\phi+\sin^2(\phi+\psi)+2(N-1)\sin\phi\sin(\phi+\psi)\\
        &=(N-1)^2+1+2(N-1)(\cos(\phi+\psi)\cos\phi+\sin(\phi+\psi)\sin\phi)\\
        &=(N-1)^2+1+2(N-1)\cos\psi.
    \end{split}\]
    \[\begin{split}
        f_2 &=(\cos\psi-1)^2+\sin^2\psi\\
        &=2(1-\cos\psi).
    \end{split}\]
    \[\begin{split}
        f_3 &=(\cos\psi-1)((N-1)\cos\phi+\cos(\phi+\psi))+\sin\psi((N-1)\sin\phi+\sin(\phi+\psi))\\
        &=(N-1)[\cos\psi\cos\phi+\sin\psi\sin\phi-\cos\phi]+\\
        &+[\cos\psi\cos(\phi+\psi)+ \sin\psi\sin(\phi+\psi)-\cos(\phi+\psi)]\\
        &=(N-1)[\cos(\phi-\psi)-\cos\phi]-[\cos(\phi+\psi)-\cos\phi].
    \end{split}
    \]
    \noindent We now return to calculating the probability $P(\ket{\tau})$.

    \[\begin{split}
        P(\ket{\tau}) &= \frac{1}{N^2}\biggl\{(\sin^2\alpha) f_1 + (N-1)(\cos^2\alpha) f_2 +(2\sin\alpha\cos\alpha\sqrt{N-1}) f_3\biggr\}\\
        &=\frac{1}{N^2}\biggl\{\sin^2\alpha[(N-1)^2+1+2(N-1)\cos\psi]+2(N-1)\cos^2\alpha(1-\cos\psi)+\\
        &+2\sin\alpha\cos\alpha\sqrt{N-1}[(N-1)(\cos(\phi-\psi)-\cos\phi)-(\cos(\phi+\psi)-\cos\phi)]\biggr\}\\
        &=\frac{1}{N^2}\biggl\{\sin^2\alpha[(N-1)^2+1]+\cos^2\alpha(2(N-1))+2(N-1)\cos\psi(\sin^2\alpha-\cos^2\alpha)+\\
        &+2\sin\alpha\cos\alpha\sqrt{N-1}[(N-1)(\cos(\phi-\psi)-\cos\phi)-(\cos(\phi+\psi)-\cos\phi)]\biggr\}\\
        &=\frac{1}{N^2}\biggl\{\sin^2\alpha(N-2)^2+2(N-1)-2(N-1)\cos(2\alpha)\cos\psi+\\
        &+\sin(2\alpha)\sqrt{N-1}[(N-1)(\cos(\phi-\psi)-\cos\phi)-(\cos(\phi+\psi)-\cos\phi)]\biggr\}.
    \end{split}\]

    Let us now examine the final line further. We can break this into two distinct parts where one is constant in $\phi$ and $\psi$, which we will denote by $c$, and the other depends on them, which we will call $g=g(\psi,\phi)$.

    \[
    c = \frac{1}{N^2}[(N-2)^2\sin^2\alpha+2(N-1)].
    \]
    \[
    \begin{split}
        g(\psi,\phi)&=\frac{1}{N^2}\biggl\{(N-1)^{3/2}\sin(2\alpha)(\cos(\phi-\psi)-\cos\phi)-2(N-1)\cos(2\alpha)\cos\psi-\\
        &-(N-1)^{1/2}\sin(2\alpha)(\cos(\phi+\psi)-\cos\phi)\biggr\}.
    \end{split}
    \]

    Therefore to find the phase changes that will maximize the probability at any given step, we need to solve the argument that maximizes the function $g$ as follows.

    \[
    \begin{split}
                \argmax_{\psi,\phi\in[-\pi,\pi]}&\ \biggl\{(N-1)^{3/2}\sin(2\alpha)(\cos(\phi-\psi)-\cos\phi) -\\
                &-2(N-1)\cos(2\alpha)\cos\psi-(N-1)^{1/2}\sin(2\alpha)(\cos(\phi+\psi)-\cos\phi)\biggr\}.
    \end{split}
    \]
    
\end{proof}

We will now take this function and analyze it from two different angles. First, we will examine it numerically to identify the optimal phase changes $\phi$ and $\psi$. Then we will approach it from an analytical point of view to better understand the objective function, and also provide insight into the behavior of the optimal $(\psi,\phi)$ when the data set size $N$ is very large.

\section{\bf Numerical Solution to Objective Function}

To analyze \cref{eq:argmaxf} numerically we utilized the global optimization toolbox in MATLAB. Since $\alpha$ is a given value from the current amplitude vector, we created a list of 1000 evenly-spaced values over $[0,\pi/2]$.  For each $\alpha$ value given, we then ran interior-point methods with randomly selected initial points for both $\psi$ and $\phi$ to find the global maximum and the corresponding optimal $\phi$ and $\psi$ values. Since the objective function also relies on $N$, we ran the analysis for three different $N$ values, $N=2^{10}$, $N=2^5$, and $N=2^4$. We then rounded the $\psi$ and $\phi$ values to the hundredth place to account for minor discrepancies coming from the optimization process. The results are graphed in \cref{fig:210phasematch}, \cref{fig:25phasematch}, and \cref{fig:24phasematch}. The x-axis is the $\phi$ phase change and the y-axis is the $\psi$ phase change. Each point represents a given $\alpha$ value and its optimal $\phi$ and $\psi$ phase changes. The labeled point is where $\phi=\psi=\pi$ i.e. classical Grover's algorithm, and is the only instance where multiple $\alpha$ values shared the same optimal $\phi$ and $\psi$ phase changes.

\begin{figure}[h!]
    \centering
    \includegraphics[width=.6\linewidth]{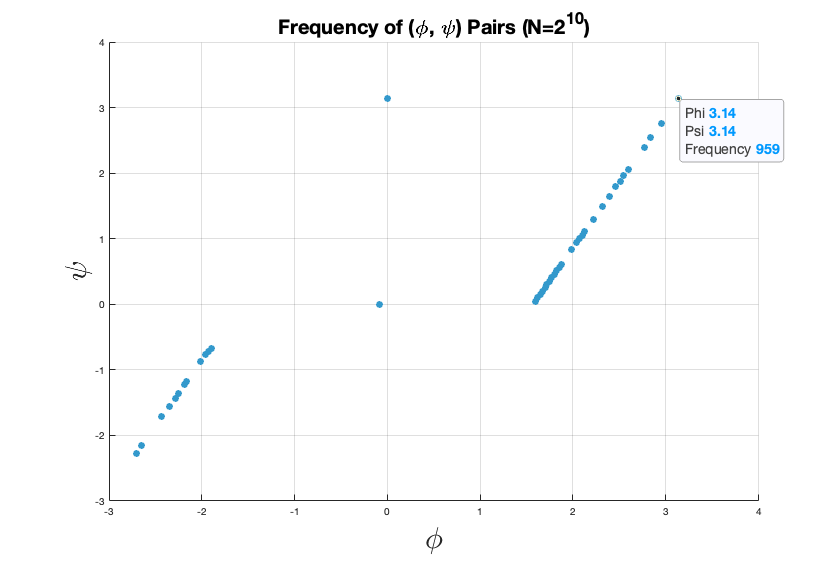}
    \caption{Graph of the optimal $(\psi,\phi)$ pairs for a set of size $N=2^{10}$. The optimal pair changes from $\psi=\phi=\pi$ at $P(\ket{\tau})=.9962$.}
    \label{fig:210phasematch}
\end{figure}

\begin{figure}[h!]
    \centering
    \includegraphics[width=.6\linewidth]{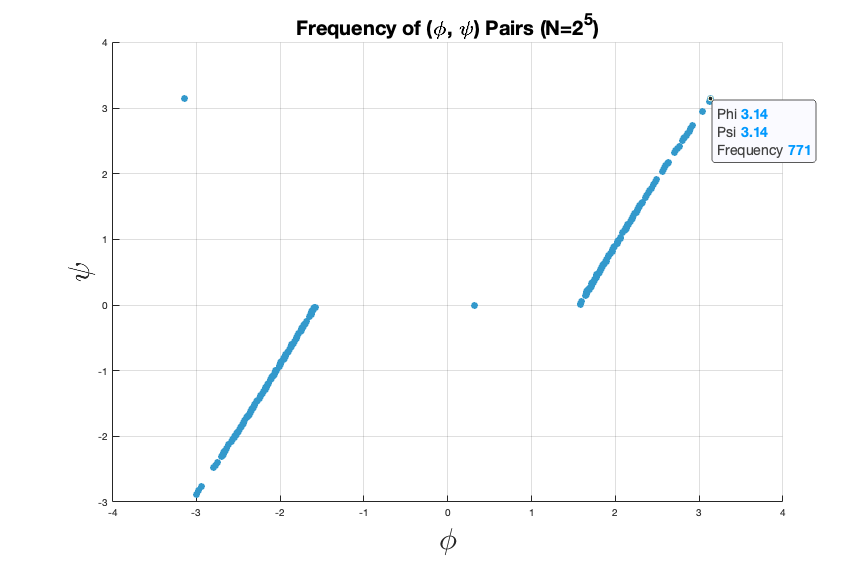}
    \caption{Graph of the optimal $(\psi,\phi)$ pairs for a set of size $N=2^{5}$. The optimal pair changes from $\psi=\phi=\pi$ at $P(\ket{\tau})=.8789$.}
    \label{fig:25phasematch}
\end{figure}

\begin{figure}[h!]
    \centering
    \includegraphics[width=.6\linewidth]{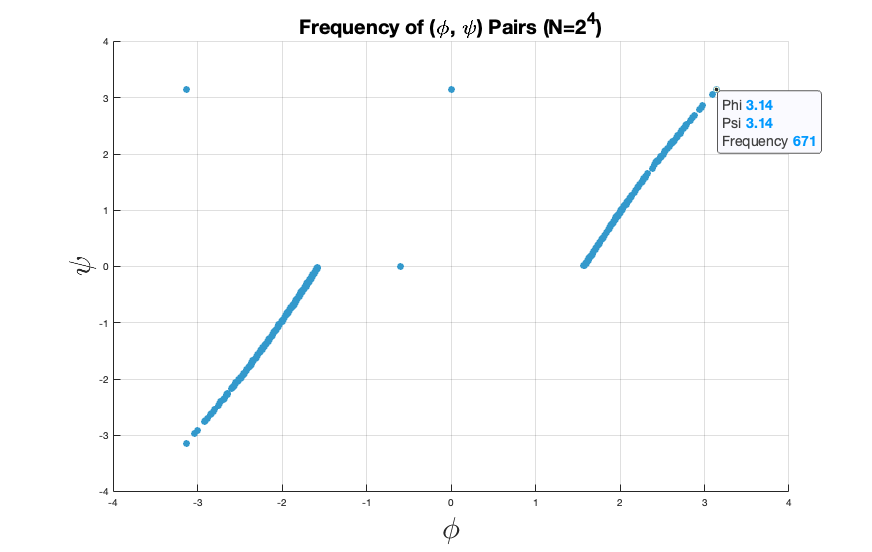}
    \caption{Graph of the optimal $(\psi,\phi)$ pairs for a set of size $N=2^{4}$. The optimal pair changes from $\psi=\phi=\pi$ at $P(\ket{\tau})=.7662$.}
    \label{fig:24phasematch}
\end{figure}

As can be seen in \cref{fig:210phasematch}, for $N=2^{10}$ the vast majority (959 of 1000) of $\alpha$ values have optimal $\phi$ and $\psi$ values of $\pi$. Investigating this further, we see that for all $\alpha$ values from $\alpha = 0$ to $\alpha = 1.5095$, there is phase matching and both equal $\pi$. This $\alpha$ value equates to a probability of observing the target state at .9962. However for larger $\alpha$ values, we no longer observe phase matching and the two instead have a non-trivial relationship. Since each point not labeled was found to be optimal only one time, we see that there is a lot of change in optimal phases as the probability of observing the target approaches one.

The results for $N=2^5$ (seen in \cref{fig:25phasematch}) and $N=2^4$ (seen in \cref{fig:24phasematch}) follow this same pattern, just with a lower cutoff for $\alpha$. For $N=2^5$, we have a cutoff at $\alpha = 1.2154$ and $P(\ket{\tau}) = .8789$, and for $N=2^4$, $\alpha = 1.0661$ and $P(\ket{\tau})=.7662$. This result does correspond with previous observations in the work by Long \cite{Lo1} and the common understanding that for smaller $N$ values, there is a substantial region in which classical Grover's algorithm is not optimal. However this does differ from that understanding because we see that phase matching is not optimal.

We would like to note here that the $\phi$ value we are optimizing for is in fact not truly the phase change $\phi$ that would be realized in a fully generalized Grover's algorithm. As mentioned in \cref{lem:phaseadj}, this $\phi$ is actually the phase change plus $\theta$, the complexity of our original amplitude vector. Therefore the phase matching we are seeing is really between $\psi$ and $\phi+\theta$. In the case that a Hadamard initialization is used for the first amplitude, then $\theta = 0$ and will remain zero after every application of the classical Grover's iterate with $\phi=\psi=\pi$. This correlates with our previous work of optimizing a generalized Grover's algorithm where we held the phase changes equal and saw that the optimal phase change was approximately $\pi-\theta$ for $\alpha$ before a specific threshold that we identified is reached \cite{CK}.

Putting this all together, we can analyze the numerical results as follows. For a set of size $N\geq 2^{10}$ and with the Hadamard initialization, classical Grover's algorithm will be sufficient in raising the probability of observing the target state until the probability gets fairly close to 1. However if there is a different initialization that results in a complex initial amplitude, then one of the phase changes will depend on that complexity, and result in non-phase matching, while the other phase change will stay $\pi$ until the probability of observing the target state gets close to 1.

\section{{\bf Analytical Observations of Objective Function}}

Now we investigate this optimization to better understand the relationship between $\psi$ and $\phi$ when optimizing the fully generalized Grover's iterate. Let us return to the optimization problem, \cref{eq:argmaxf}. We wish to analyze this statement by way of its leading term, and to enable this, let us assume that $N$ is sufficiently large. There are now two different ways that we can start the analysis. First in \cref{th:hadamardinit} we will assume that the state is initialized with the Hadamard gate, and find the optimal phase change for the first iteration of a fully generalized Grover's algorithm. Next in \cref{th:po1}, we assume that there have been some amount of iterations such that $P(\ket{\tau})$ stays away from zero and also is not too close to one. To formalize this assumption, we consider the case where there exists a function $f(N)$ such that $\lim_{N\to\infty}f(N)=\infty$ and $P(\ket{\tau})=\frac{f(N)}{N}$.

\begin{theorem}\label{th:hadamardinit}
    Assume that $N$ is sufficiently large and the amplitude vector is initialized with the Hadamard gate. Then the first iteration of a fully generalized Grover's algorithm will have an optimal phase change of $\phi=\psi=\pi$.
\end{theorem}

\begin{proof}
    Recall from \cref{eq:hadamardinitvec} that if the system is initialized with the Hadamard gate, we know that $\sin\alpha = \frac{1}{\sqrt{N}}$ and $\cos\alpha = \sqrt{\frac{N-1}{N}}$ . Therefore if we return to \cref{eq:argmaxf}, we can simplify it with this initial condition, and those gained from using trig identities.

\[\begin{split}
    \max_{\phi,\psi}f(\phi,\psi) &= \max_{\phi,\psi}\biggl\{(N-1)^{3/2}\sin(2\alpha)(\cos(\phi-\psi)-\cos\phi)-2(N-1)\cos(2\alpha)\cos\psi-\\
    &-(N-1)^{1/2}\sin(2\alpha)(\cos(\phi+\psi)-\cos\phi)\biggl\}
\end{split}\]
\[\begin{split}
    &=\max_{\phi,\psi}\biggl\{\frac{2(N-1)^2}{N}(\cos(\phi-\psi)-\cos\phi)-\frac{2(N-1)(N-2)}{N}\cos\psi-\\
    &-\frac{2(N-1)}{N}(\cos(\phi+\psi)-\cos\phi)\biggr\}.
\end{split}\]

Since we assume that $N$ is sufficiently large, we may as well focus on the leading order term given by

\[\cos(\phi-\psi)-\cos\phi-\cos\psi = h(\phi,\psi),\]

\noindent since this is the term with coefficients order $O(N)$. With $h(\phi,\psi)$ in hand, we can now perform a critical point analysis of it. We get that  $\nabla h = \langle-\sin(\phi-\psi)+\sin\phi,\ \sin(\phi-\psi)+\sin\psi\rangle$. To get critical points, we observe that

\[\begin{cases}
    \sin(\phi-\psi)=\sin\phi\\
    \sin(\phi-\psi)=-\sin\psi
\end{cases}\Rightarrow\sin\phi = -\sin\psi\]

Using this relationship, we get that

\[\begin{split}
    \sin\phi &= \sin(\phi-\psi) = \sin\phi\cos\psi-\cos\phi\sin\psi\\
    &= \sin\phi\cos\psi+\cos\phi\sin\phi,
\end{split}\]

\noindent which is equivalent to

\[0=\sin\phi(\cos\psi+\cos\phi-1).\]

From $\sin\phi = 0$ we get the critical points $(0,0), (0,\pi), (\pi,0),$ and $(\pi,\pi)$. From $\cos\psi+\cos\phi-1=0$ we get that $\cos\psi = 1-\cos\phi$, and since $\sin\psi=-\sin\phi$, we know that 

\[(1-\cos\phi)^2+(-\sin\phi)^2=1\]
\[\Rightarrow 1+\cos^2\phi-2\cos\phi+\sin^2\phi = 1\]
\[\Rightarrow \frac{1}{2}=\cos\phi.\]

Using this, we get that $\phi = \frac{\pi}{3}$ or $\phi=\frac{5\pi}{3}$. Since $\sin\psi = -\sin\phi$, we get the critical points of $(\frac{\pi}{3},\frac{5\pi}{3})$ and $(\frac{5\pi}{3},\frac{\pi}{3}).$ In total we have 6 critical points identified, and we will use the Hessian to classify them.

Recall that $\nabla h = \langle-\sin(\phi-\psi)+\sin\phi,\ \sin(\phi-\psi)+\sin\psi\rangle$.

\[H_h(\phi,\psi) = \begin{bmatrix}
    -\cos(\phi-\psi)+\cos\phi & \cos(\phi-\psi)\\
    \cos(\phi-\psi) & -\cos(\phi-\psi)+\cos\psi
\end{bmatrix}\]
\[H_h(0,0) = \begin{bmatrix}
    0 & 1\\
    1 & 0
\end{bmatrix}\quad H_h(0,\pi) = \begin{bmatrix}
    2 & -1\\
    -1 & 0
\end{bmatrix}\quad H_h(\pi,0)=\begin{bmatrix}
    0 & -1\\
    -1 & 2
\end{bmatrix}\quad H_h(\pi,\pi) = \begin{bmatrix}
    -2 & 1\\
    1 & -2
\end{bmatrix}\]

\[H_h\biggl(\frac{\pi}{3},\frac{5\pi}{3}\biggr)=\begin{bmatrix}
    1 & \frac{-1}{2}\\
    \frac{-1}{2} & 1
\end{bmatrix}\quad H_h\biggl(\frac{5\pi}{3},\frac{\pi}{3}\biggr)=\begin{bmatrix}
    1 & \frac{-1}{2}\\
    \frac{-1}{2} & 1
\end{bmatrix}\]

Analyzing these Hessian matrices, we find that $(0,0), (0,\pi),$ and $(\pi,0)$ are saddle points, $(\frac{\pi}{3}, \frac{5\pi}{3})$ and $(\frac{5\pi}{3},\frac{\pi}{3})$ are local minima, and $(\pi,\pi)$ is a local maximum. 
\end{proof}

To discuss the findings from \cref{th:hadamardinit}, let us start by examining the local minima at $(\frac{\pi}{3}, \frac{5\pi}{3})$ and $(\frac{5\pi}{3},\frac{\pi}{3})$. If $N$ is large enough and the qubits are initialized using the Hadamard gate, these phase change pairs lead to the worst performance of the first iterate in the fully generalized Grover's algorithm. If we change our attention to the local maximum, we find the phase change pair $(\pi,\pi)$. So in this case, if $N$ is large enough and the qubits are initialized by the Hadamard gate, the best choice for the phase changes is phase matching with $\phi=\psi=\pi$ which is the classical Grover's algorithm.

\begin{theorem}\label{th:po1}
    Assume that $N$ is sufficiently large and there exists a function $f(N)$ such that $\lim_{N\to\infty}f(N)=\infty$ and $P(\ket{\tau})=\frac{f(N)}{N}$. Then the optimal phase change for a fully generalized Grover's algorithm is $\phi=\psi=\pi$.
\end{theorem}

\begin{proof}
    Since $P(\ket{\tau})=\frac{f(N)}{N}$, then we know that $\sin\alpha = \sqrt{P(\ket{\tau})} = \sqrt{\frac{f(N)}{N}}$ before the update. Let us now look at \cref{eq:argmaxf} under our assumptions.

\begin{equation}
    \begin{split}
                \argmax_{\psi,\phi\in[-\pi,\pi]}&\ \biggl\{(N-1)^{3/2}\sin(2\alpha)(\cos(\phi-\psi)-\cos\phi)-2(N-1)\cos(2\alpha)\cos\psi\\
                & -(N-1)^{1/2}\sin(2\alpha)(\cos(\phi+\psi)-\cos\phi)\biggr\}.
    \end{split}
\end{equation}

Here we can see that since $N$ is sufficiently large and $\sin\alpha =\sqrt{\frac{f(N)}{N}}$ that the leading term will be $(N-1)^{3/2}\sin(2\alpha)(\cos(\phi-\psi)-\cos\phi)$. Focusing on this specifically, if we are interested in maximizing over $\phi$ and $\psi$, we can simply focus on $\cos(\phi-\psi)-\cos(\phi)$. We can now just rely on a simple analysis of the gradient to identify all the critical points. Let $g(\phi,\psi) = \cos(\phi-\psi)-\cos\phi$. Then $\nabla g = \langle -\sin(\phi-\psi)+\sin\phi,\ \sin(\phi-\psi)\rangle$, and our critical points are $\{(0,0), (0,\pi), (\pi,0), (\pi,\pi)\}$. We then use the Hessian to classify these points.

\[H_g(\phi,\psi) = \begin{bmatrix}
    g_{\phi\phi} & g_{\phi\psi}\\
    g_{\psi\phi} & g_{\psi\psi}
\end{bmatrix}=\begin{bmatrix}
    -\cos(\phi-\psi)+\cos\phi & \cos(\phi-\psi)\\
    \cos(\phi-\psi) & -\cos(\phi-\psi)
\end{bmatrix}\]
\[H_g(0,0) = \begin{bmatrix}
    0 & 1\\
    1 & -1
\end{bmatrix}\quad H_g(0,\pi) = \begin{bmatrix}
    2 & -1\\
    -1 & 1
\end{bmatrix}\quad H_g(\pi,0) = \begin{bmatrix}
    0 & -1\\
    -1 & 1
\end{bmatrix}\quad H_g(\pi,\pi) = \begin{bmatrix}
    -2 & 1\\
    1 & -1
\end{bmatrix}\]

Using the second derivative test, we see that $(0,0)$ and $(\pi,0)$ are saddle points, $(0,\pi)$ is a local min, and $(\pi,\pi)$ is a local max. What this means is that in the case where $N$ is sufficiently large and there exists a function $f(N)$ such that $\lim_{N\to\infty}f(N)=\infty$ and $P(\ket{\tau})=\frac{f(N)}{N}$, then the optimal phase change occurs at $\phi=\psi=\pi$ which is the traditional Grover's algorithm. 
\end{proof}

To make a quick remark on the assumption that $N$ is sufficiently large, the ability for us to focus on the leading term happens even in relatively small systems due to the fact that $N=2^n$ where $n$ is the number of qubits. With this in mind, let us keep our assumption on $P(\ket{\tau})$ and that $n=10$. We can first look at the $(N-1)$ coefficients in \cref{eq:argmaxf}, since the rest of each term is a product and sum of trig functions which by our assumption will not be too small. Then we can see that $N=2^{10}$ will lead to the term with the coefficient $(N-1)^{3/2}$ overpowering the rest and being the main piece we can focus our analysis on.

We would also like to draw an analogous argument from this analysis and most notably \cref{th:po1} to our previous work on optimizing generalized Grover's algorithm \cite{CK}. What we found previously was that as long as $P(\ket{\tau})$ stayed away from zero and one, the optimal phase change (under the assumption that $\phi=\psi$) was numerically shown to be approximately $\pi$ given that the amplitude vector was real. If it wasn't real, then it was the case that the optimal phase change was approximately equal to $\pi-\theta$ where $\theta$ is from $e^{i\theta}$ in the description of the amplitude vector's complexity. Also we have proved that the optimal phase is precisely $\pi$ until the target probability reaches a threshold quite close to one, which is also identified in the paper, in case we begin with the Hadamard initial gate. This coincides with the result we are seeing now where the optimal phase change is $\pi$ as long as $P(\ket{\tau})$ stays away from zero and one.

\section{\bf Analytical Observations - Testing Special Cases}

Now instead of attempting to identify all critical points of our objective function \cref{eq:argmaxf}, let us examine two notable cases of local maximums from this and our previous analysis. These two are where classical Grover's is optimal, $\phi=\psi=\pi$, and where no change is optimal, $\phi=\psi=0$, which we can think of as the algorithm getting stuck at a local maximum. It is immediately apparent that these are critical values for any $N$ and $\alpha$, so we will focus on identifying what they are. First, let us calculate the Hessian of \cref{eq:argmaxf}. To assist in simplifying notation, let $M=N-1$.

\[H(\phi,\psi) = \begin{bmatrix}
    H_{11}(\phi,\psi) & H_{12}(\phi,\psi)\\
    H_{21}(\phi,\psi) & H_{22}(\phi,\psi) 
\end{bmatrix},\]

\noindent where

\[H_{11}(\phi,\psi) = \sin(2\alpha)[M^{3/2}(\cos\phi-\cos(\phi-\psi))-M^{1/2}(\cos\phi-\cos(\phi+\psi))]\]
\[H_{12}(\phi,\psi) = H_{21}(\phi,\psi) = \sin(2\alpha)[M^{3/2}\cos(\phi-\psi)+M^{1/2}\cos(\phi+\psi)]\]
\[H_{22}(\phi,\psi) = -M^{3/2}\sin(2\alpha)\cos(\phi-\psi)+2M\cos(2\alpha)\cos\psi+M^{1/2}\sin(2\alpha)\cos(\phi+\psi).\]

Now let us analyze $H(\phi,\psi)$ for our two cases.

\textbf{Case 1: $\phi=\psi=\pi$}

First we let $\phi=\psi=\pi$ which corresponds to classical Grover's. We wish to analyze when this is a local maximum, so we will first look at the Hessian. We have that the entries of our Hessian matrix are

\[H_{11}(\pi,\pi) = -2\sin(2\alpha)(M^{3/2}-M^{1/2})\]
\[H_{12}(\pi,\pi)=H_{21}(\pi,\pi)= \sin(2\alpha)(M^{3/2}+M^{1/2})\]
\[H_{22}(\pi,\pi) = \sin(2\alpha)(M^{1/2}-M^{3/2})-2M\cos(2\alpha).\]

For this to be a local maximum, then the following statements must be true.

\[\begin{cases}
    \textbf{(1)}\ H_{11}(\pi,\pi) < 0\\
    \textbf{(2)}\ H_{22}(\pi,\pi)<0\\
    \textbf{(3)}\ H_{11}(\pi,\pi)H_{22}(\pi,\pi)-H_{12}(\pi,\pi)^2 > 0
\end{cases}\]

First we can directly see that \textbf{(1)} is true if $M\geq 1$. This is a reasonable assumption to make since $M=N-1$ where $N$ is the size of the system. Next, we can see that if $M\geq 1$ and $\textbf{(3)}$ is true, then \textbf{(2)} must be true as well. Then to determine when classical Grover's is optimal, we need only find for what values of $\alpha$ and $M$ give us $H_{11}(\pi,\pi)H_{22}(\pi,\pi)-H_{12}(\pi,\pi)^2 >0$. By expanding this, we get the new criteria of

\[\sin^2(2\alpha)(M^3-6M^2+M)+4M\sin(2\alpha)\cos(2\alpha)(M^{3/2}-M^{1/2})>0.\]

This is graphed in \cref{fig:nspmgroveroptim} to illustrate where classical Grover's is optimal.

\begin{figure}
    \centering
    \includegraphics[width=0.75\linewidth]{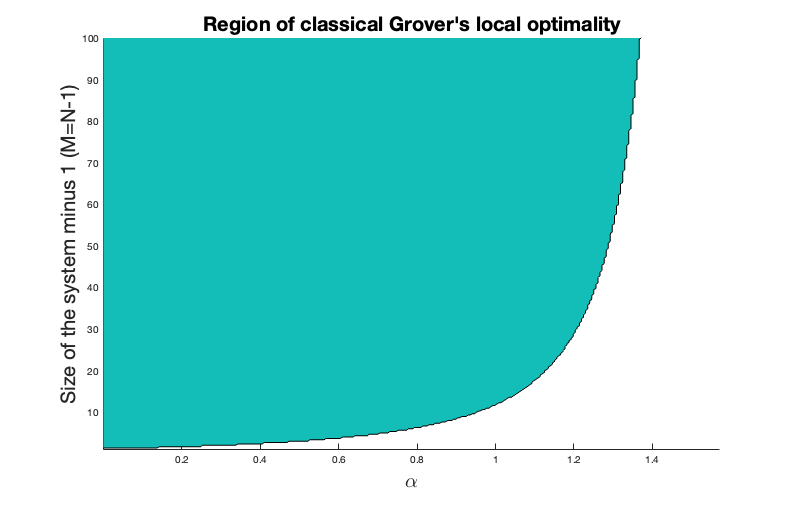}
    \caption{The shaded region is the set of $\alpha$ and $M$ values for which classical Grover's algorithm is locally optimal. The non-shaded region is consequently where an improvement can be made by appropriately choosing phase changes.}
    \label{fig:nspmgroveroptim}
\end{figure}

With this we see the same pattern of behavior where until the value of $\alpha$ gets close to $\pi/2$, and therefore $P(\ket{\tau})$ gets close to one, classical Grover's algorithm is locally optimal. However this also means that in the non-shaded region, a pair of phase changes that differ from $\pi$ will lead to a better increase in the probability of observing the target after one step. To determine the locally optimal pair, the objective statement \cref{eq:argmaxf} can be solved using numerical techniques.

\textbf{Case 2: $\phi=\psi=0$}

Our interest in this case comes from the problem of overshooting the target in classical Grover's algorithm. If more than the optimal number of steps are used, then the probability of observing the target will begin to decrease. In our previous work, this manifested as there existing a region when $\alpha$ is close to $\pi/2$ where the locally optimal pair of phase changes was zero \cite{CK}. So now we wish to determine if the algorithm can still get stuck in a region where the local optimal is to do no change whatsoever. By setting $\phi=\psi=0$ we can immediately see that $H_{11}(0,0) = 0$. This means that the critical point $\phi=\psi=0$ will always be a saddle point, and hence not a local maximum. In this context then, it means that when strict phase matching is not enforced, the algorithm optimized locally will not get stuck in a region where it is best to no longer make any changes.

\section{\bf Example Case when $N$ is not Large}

Let us now consider an example scenario in which our analysis can see improvements compared to traditional Grover's algorithm. Let $N=2^5$, and let the initial amplitude vector be initialized using the Hadamard gate, meaning it is entirely real. We will compare classical Grover's algorithm with phase changes of $\psi=\phi=\pi$ and phase matching to the fully generalized Grover's algorithm with phase changes that come from optimal choices of $\psi$ and $\phi$ according to \cref{eq:argmaxf} and compare the target probabilities.

\begin{table}[h!]
\centering
\begin{tabular}{ccccc}\toprule
    & \multicolumn{2}{c}{Target} & \multicolumn{2}{c}{Phase Changes} \\ \cmidrule(lr){2-3}\cmidrule(l){4-5}
     Step & Amplitude & Probability & $\phi$ & $\psi$\\ \midrule 
    1 & 0.5082 & 0.2583 & $\pi$ & $\pi$\\
    2 & 0.7762 & 0.6024 & $\pi$ & $\pi$\\
    3 & 0.9471 & 0.8969 & $\pi$ & $\pi$\\
    4 & 0.9996 & 0.9992 & $\pi$ & $\pi$\\ \bottomrule
\end{tabular}
\centering
\vspace{3mm}
\caption{Classical Grover's algorithm used with a Hadamard initialized amplitude vector}
\label{tab:classicgrover}
\end{table}

\begin{table}[ht!]
\centering
    \begin{tabular}{ccccc}\toprule
    & \multicolumn{2}{c}{Target} & \multicolumn{2}{c}{Phase Changes} \\ \cmidrule(lr){2-3}\cmidrule(l){4-5}
    Step & Amplitude & Probability & $\phi$ & $\psi$\\ \midrule 
    1 & 0.5082 & 0.2583 & $\pi$ & $\pi$\\
    2 & 0.7762 & 0.6024 & $\pi$ & $\pi$\\
    3 & 0.9471 & 0.8969 & $\pi$ & $\pi$\\
    4 & 0.9226 + 0.3858i & $\approx 1.0$ & -2.7218 & -2.3493\\ \bottomrule
\end{tabular}
\centering
\vspace{3mm}
\caption{Fully generalized Grover's algorithm with optimized $\phi$ and $\psi$ using \cref{eq:argmaxf} with a Hadamard initialized amplitude vector. For the fourth step, the probability differs from 1 by $3.7252\times 10^{-10}$.}
\label{tab:fggrover}
\end{table}

As can be seen in \cref{tab:classicgrover} and \cref{tab:fggrover}, the target element reaches its highest probability of observation on the fourth step. For classical Grover's we get a probability of .9992 and for the fully generalized Grover's a probability that differs from 1 by $3.7252\times 10^{-10}$. While the increase in probability is marginal, there is a clear improvement on the theoretical observation rate when using the optimization formula to determine the phase changes used in the fully generalized Grover's algorithm.

\section{{\bf Conclusion \& Further Research}}

We have found that in a fully generalized Grover's algorithm and for a majority of the iterations, phase matching is optimal, and in fact, the optimal phase change is $\phi=\psi=\pi$ which corresponds to the classical Grover's algorithm. We showed this numerically by finding where the optimal phase changes will differ from $\pi$ for sets of size $N=2^4$, $N=2^5$, and $N=2^{10}$ with the cut-off probability increasing with $N$. Analytically we took two different approaches. The first was finding and analyzing the critical points under reasonable assumptions, which found that classical Grover's is a local maximum. In the second approach we analyzed the critical points $(\pi,\pi)$ and $(0,0)$. For the critical point of $(\pi,\pi)$ we found explicit inequalities for which it will no longer be a local maximum. For the critical point of $(0,0)$ we found that this will never be a local maximum and will instead be a saddle point, which differs from the case of strict phase matching where $\phi=\psi$. An important caveat to all of this is that if the amplitude vector is complex, then the phase matching will actually be between $\psi$ and $\phi+\theta$ where $\theta$ is the phase of the amplitude vector.

Looking towards our future research, our main interest lies in quantum random walks, and their applications to searching quantum spatial data. As stochastic objects, they exhibit interesting behavior that diverges from what we have observed in classical random walks. As for their use in search of spatial data, we feel there is promise in this area due to their applications in classical computing for searching large datasets. This is also an active area of research that we have begun to study. We look forward to continuing our work in quantum search algorithms, as well as analyzing stochastic objects in the quantum computation setting.

\bibliography{sn-bibliography}
\bibliographystyle{unsrt}

\end{document}